\documentclass[10pt,a4paper]{article}
\usepackage{jheppub}
\usepackage{pdflscape}
\usepackage{amsmath}
\usepackage{amssymb}
\usepackage{amsthm}
\usepackage{tikz-feynman}
\usepackage{dcolumn}
\usepackage{bm}
\usepackage{epsfig}
\usepackage{amsfonts}
\usepackage{graphicx}
\usepackage{subfigure}
\usepackage{dcolumn}

\newtheorem{theorem}{Theorem}

\begin{document}
\title{ Cylinder quantum field theories at small coupling}

\author[a]{Andrei Ioan Dogaru,}
\author[b]{Ruben Campos Delgado}

\affiliation[a]{Faculty of Physics, University of Bucharest, \\Strada Atomiștilor 405, 077125 Magurele, Romania}
\affiliation[b]{Bethe Center for Theoretical Physics and Physikalisches Institut der Universit\"at Bonn,\\
Nussallee 12, 53115 Bonn, Germany}

\emailAdd{andreidogaruioan95@gmail.com}

\emailAdd{ruben.camposdelgado@gmail.com}

\abstract{
We show that any 2D scalar field theory compactified on a cylinder and with a Fourier expandable potential $V$ is equivalent, in the small coupling limit, to a 1D theory involving a massless particle in a potential $V$ and an infinite tower of free massive Kaluza-Klein (KK) modes.  Moving slightly away from the deep IR region has the effect of
switching on interactions between the zero mode and the KK modes, whose strength is controlled by powers of the coupling, hence making the interactions increasingly suppressed. We take the notable example of Liouville field theory and, starting from its worldline version, we compute the torus (one-loop) partition function perturbatively in the coupling constant. The partition function at leading order is invariant under a T-duality transformation that maps the radius of the cylinder to its inverse and rescales it by the square of the Schwinger parameter of the cylinder. We show that this behavior is a universal feature of cylinder QFTs.}


\maketitle
%%%%%%%%%%%%%%%%%%%%%%%%
\section{Introduction}
The worldline formulation of quantum field theory (QFT) is an important tool for organising perturbation theory \cite{Strassler:1992zr}, as well as for describing non-perturbative effects \cite{Antonov:2016rzr}. In this framework, the usual Feynman graphs of a QFT are reinterpreted as terms in the path integral of a quantum mechanics describing the propagation of corresponding particles. This approach allows one to regard a $(d+1)-$dimensional QFT as an ordinary, albeit complicated, quantum mechanical system: a $(0+1)-$dimensional QFT. For a review about the worldline formalism, see for example \cite{Schubert:2001he}. Another powerful idea in QFT is to relate different theories via Kaluza-Klein (KK) compactification \cite{kkreview}. This paper is an attempt to utilise both techniques, by linking them against several $(1+1)-$dimensional QFTs. 
The first point of convergence of these two ideas is to be found in string theory, since the latter is both a non-local generalisation of a worldline theory and a natural client of the compactification technology. That said, one usually encounters \textit{spacetime} compactifications in string theory, which underlie most cases of duality and are considered a key ingredient in obtaining realistic models of high-energy physics \cite{Grana:2005jc}. However, in \cite{Abel:2020gdi} string theoretic features such as non-locality and modular invariance are shown to emerge from a worldline action with an infinite tower of massive KK modes. This action is obtained directly from the Polyakov action by means of a \textit{worldsheet} compactification and is used to compute the one-loop partition function. 

Inspired by these results, we study the cylinder compactification of 2D field theories at small coupling, turning them into worldline theories. The main motivation for this work is to create a dictionary that allows one to view a given theory by two different perspectives. One can deal with a 2D field theory or, alternatively, one can translate it into a 1D theory of particles. In the paper we argue that any two-dimensional field theory defined on a cylinder $\Sigma$, equipped with a coupling constant $\epsilon$ and with a Lagrangian description of the form
\begin{equation}
S=\frac{1}{4\pi}\int_{\Sigma}d^2\sigma \sqrt{g}\,\left(g^{ab}\partial_a X\partial_b X+4\pi V_\epsilon(X)\right)
\end{equation}
is equivalent, in the $\epsilon\to 0$ limit, to a one-dimensional theory involving a massless particle $X_0$ satisfying the Schr\"odinger equation in a potential $V_\epsilon$ plus an infinite tower of free massive particles $X_n$, the Kaluza-Klein modes. Moving slightly away from the deep IR region has the effect of switching on interactions between $X_0$ and $X_n$, whose strength is controlled by powers of $\epsilon$. As a consequence of that, they are increasingly suppressed. 
We illustrate our procedure by choosing a 2D field theory whose properties are well known, namely Liouville field theory, which is characterised by an exponential potential. 
We put the theory on a cylinder by making the spatial coordinate of the worldsheet periodic. A Laurent-expansion method is then deployed to deal with the exponential term.
We advocate the equivalence of the resulting one-dimensional theory with the parent two-dimensional theory by computing the one-loop partition function with an additionally compactified target space (double compactification). Our result recovers the one obtained in the literature using matrix models techniques \cite{Gross:1990ub, Nakayama:2004vk}. The partition function is computed on a torus resulting from deforming the cylinder (hence it is considered to be "one-loop") and is obtained by integrating over all propagators for $X_0$ and $X_n$. At leading order, it presents a specific case of symmetry transformation, or T-duality, which maps the radius of the cylinder to its inverse and rescales it by some constant, the square of the Schwinger parameter of the cylinder. With the partition function at our disposal, we also compute general one-loop scattering amplitudes by integrating it with suitable worldline vertex operators. 

The whole procedure can be generalised to other potentials, as soon as their Fourier transform exists. We discuss this by taking as a further example a field theory with a one-soliton, or P\"oschl-Teller, potential. We compute the one-loop partition function for this new theory and find that it has the same symmetry as the Liouville one. In fact, this is a universal feature. We show at the end of the paper that, given any 2D scalar field theory with a Fourier-expandable potential and compactified on a cylinder, the one-loop partition function obtained from the corresponding worldline theory presents at leading order in the coupling the same T-duality. 

The paper is organised as follows. In Section \ref{sec:section1}, we compactify Liouville field theory on a cylinder and compute the one-loop partition function perturbatively in the low coupling regime, within the worldline framework. We then provide examples of scattering amplitudes that can be computed from the partition function.  In Section \ref{sec:section2}, the same method is applied to another field theory, the two-dimensional generalisation of quantum mechanics with a one-soliton potential. We find that the new partion function has the same symmetry as the Liouville one. In Section \ref{sec:section3} we argue that this is a universal feature independent of the potential. We then end with a summary and outlook.  
%%%%%%%%%%%%%%%%%%%%%%%%%%%%%%%%%
%%%%%%%%%%%%%%%%%%%%%%%%%%%%%%%%%
\section{A worldline perspective on Liouville theory}\label{sec:section1}
In this section we calculate the one-loop partition function for Liouville field theory. We first put the theory on a cylinder by compactifying one spatial direction on a circle of radius $r\in\mathbb{R}$.  The one-dimensional theory so obtained describes one massless particle in a potential, an infinite tower of massive particles, and their interactions.  We then sew the cylinder into a torus and integrate over all regularised kernels to obtain the partition function as a perturbative expansion in the coupling constant. The  final result presents, at leading order, a particular T-duality.
%%%%%%%%%%%%%%%%%%%%%%%%%%%%%%%%%
%%%%%%%%%%%%%%%%%%%%%%%%%%%%%%%%%
\subsection{Liouville theory on a cylinder}\label{sec:liouville_cylinder}
Liouville field theory is an exactly solvable 2D CFT model described by the action
\cite{Polyakov:1981rd, Seiberg:1990eb, Nakayama:2004vk}
\begin{equation}
    S=\frac{1}{4\pi}\int_\Sigma d^2\sigma \sqrt{g}\,\left(g^{ab}\partial_{a}X\partial_{b}X+4\pi\mu e^{\beta X}+QRX\right),
\end{equation}
where $\mu$ is the cosmological constant, $\beta$ is the coupling constant of the theory, $Q=\beta/2+2/\beta$, $\Sigma$ is a Riemann surface parametrised by $\sigma_a$, $a=1,2$, with Ricci curvature $R$ and equipped with a Euclidean metric $g$. The field $X$ can be thought of as a map $X:\Sigma\to M$, where $M$ is some target space. 
In this paper we are interested in the case where $\Sigma$ is a cylinder. To parametrise it, we let $\sigma_2\in[0,1]$ and compactify the coordinate $\sigma_1$ on a circle of radius $r$ so that $\sigma_1\in[0,2\pi r]$. We then perform a Kaluza-Klein dimensional reduction via the Ansatz
\begin{equation}
    X(\sigma_1,\sigma_2)=\sum_{n\in \mathbb{Z}}X_n(\sigma_2)e^{ \frac{i n\sigma_1}{r}},
\end{equation}
interpreting $\sigma_2\equiv\tau$ as the worldline coordinate. 
In addition, we impose the reality condition $X_{-n}=X^{\dagger}_n$ and gauge fix the metric to \cite{Abel:2020gdi}
\begin{equation}
    g_{ab}=\begin{pmatrix}1 & A_{\tau} \\ A_{\tau} & g_{\tau\tau} + A^2_{\tau}
    \end{pmatrix}.
\end{equation}
$A=A_\tau d\tau$ can be interpreted as a gauge field which measures how much the cylinder is twisted. For an untwisted cylinder, $A=0$. Diffeomorphism and Weyl invariance can be fully exploited to set $A=0$ and $g_{\tau\tau}=T^2$, where $T$ is the Teichmüller parameter of the cylinder, also known as Schwinger parameter. It has the dimension of a length. 

After dimensional reduction, the free part of the action becomes
\begin{equation}
    \frac{1}{4\pi}\int_\Sigma d^2\sigma \sqrt{g}\, g^{ab}\partial_{a}X\partial_{b}X=\int_{0}^1 d\tau\,\left\{\frac{r}{2T}\dot{X}_0^2+\sum_{n=1}^{+\infty}\left(\frac{r}{T}\lvert\dot{X}_{n}\rvert^2
    +\frac{T}{r}n^2\lvert X_n\rvert^2\right)\right\}.
\end{equation}
The exponential term of the action requires more care. Firstly, it is convenient to make the change of variable $z=e^{i \sigma_1/r}$. Thus, 
\begin{equation}
    \int_0^{2\pi r} d\sigma_1\, e^{\beta X}=e^{\beta X_0}\int_0^{2\pi r}d\sigma_1\,e^{\beta\sum_{n\neq 0}X_ne^{i n\sigma_1/r}}
    =-ir e^{\beta X_0}\oint_\gamma\frac{dz}{z}e^{\beta\sum_{n=1}^{+\infty }X_n z^n}e^{\beta\sum_{n=1}^{+\infty}\frac{X_{-n}}{z^n}},
\end{equation}
$\gamma$ being the unit circle. The second exponential can now be expanded in powers of $1/z$ by means of a Laurent series. We illustrate this by explicitly keeping terms up to order $1/z^3$. We have
\begin{equation}
\begin{gathered}
    \exp\left(\beta X_{-1}\frac{1}{z}+\beta X_{-2}\frac{1}{z^2}+\cdots\right)
    =1+\beta\left(X_{-1}\frac{1}{z}+X_{-2}\frac{1}{z^2}+ X_{-3}\frac{1}{z^3}+\cdots\right)\\
    +\frac{\beta^2}{2}\left(X_{-1}\frac{1}{z}+X_{-2}\frac{1}{z^2}+ X_{-3}\frac{1}{z^3}+\cdots\right)^2
    +\frac{\beta^3}{3!}\left(X_{-1}\frac{1}{z}+X_{-2}\frac{1}{z^2}+ X_{-3}\frac{1}{z^3}+\cdots\right)^3+\cdots\\
    =a_0+\frac{a_1}{z}+\frac{a_2}{z^2}+\frac{a_3}{z^3}+\cdots,
\end{gathered}
\end{equation}
where 
\begin{equation}\label{eq:coefficients_ai}
\begin{gathered}
    a_0=1,\\
    a_1=\beta X_{-1},\\
    a_2=\beta X_{-2}+\frac{1}{2}\beta^2X^2_{-1},\\
    a_3=\beta X_{-3}+\beta^2X_{-1}X_{-2}+\frac{1}{6}\beta^3X^3_{-1}.
\end{gathered}
\end{equation}
Hence,
\begin{equation}
\begin{gathered}\label{eq:applying_cauchy}
    \int_0^{2\pi r} d\sigma_1\, e^{\beta X}=-ire^{\beta X_0}\sum_{n=0}^{+\infty} a_n\oint_{\gamma}dz \frac{f(z)}{z^{n+1}}=(2\pi r) e^{\beta X_0}\sum_{n=0}^{+\infty}\frac{a_n}{n!}f^{(n)}(z)\rvert_{z=0},
\end{gathered}
\end{equation}
where $f(z)$ is the holomorphic function 
\begin{equation}
    f(z)=\exp\left(\beta\sum_{n=1}^{+\infty}X_n z^n\right).
\end{equation}
In going from the first to the last last line of \eqref{eq:applying_cauchy} we have applied the Cauchy integral formula. 

The upshot is that the dimensional reduction of the exponential factor involves an infinite number of terms, which can grouped by powers of $\beta$. Plugging the coefficients \eqref{eq:coefficients_ai} into \eqref{eq:applying_cauchy} we find
\begin{equation}
\begin{gathered}
    \int_0^{2\pi r} d\sigma_1\, e^{\beta X}=(2\pi r)e^{\beta X_0}\bigg\{1+\beta^2\big(X_1X_{-1}+X_2X_{-2}+X_{3}X_{-3}+\cdots\big)+\beta^3\Big(\frac{1}{2}X^2_1X_{-2}\\+\frac{1}{2}X^2_{-1}X_2+X_1X_2X_{-3}+X_{-1}X_{-2}X_3+\cdots\Big)+\mathcal{O}(\beta^4)\bigg\}.
\end{gathered}
\end{equation}
We now make the assumption of small coupling. To be precise, we truncate the series up to order $\mathcal{O}(\beta^2)$:
\begin{equation}\label{eq:integral_exponential}
    \int_0^{2\pi r} d\sigma_1\, e^{\beta X}=(2\pi r)e^{\beta X_0}\left(1+\beta^2\sum_{n=1}^{+\infty}\lvert X_n\rvert^2+\mathcal{O}(\beta^3)\right).
\end{equation}
The preliminary form of the Liouville action after dimensional reduction is 
\begin{equation}\label{eq:action_compactified}
\begin{gathered}
    S=\int_0^1 d\tau \left(\frac{r}{2 T}\dot{X}_0^2+2\pi \mu r T e^{\beta X_0}\right)
    +\sum_{n=1}^{+\infty}\int_{0}^1 d\tau\, \left(\frac{r}{T}\lvert \dot{X}_n\rvert^2 +\frac{T}{r}n^2\lvert X_n\rvert^2\right)\\
    +2\pi \mu r T\beta^2\sum_{n=1}^{+\infty}\int_0^1 d\tau\, e^{\beta X_0}\lvert X_n\rvert^2.
\end{gathered}
\end{equation}
The last line of \eqref{eq:action_compactified} can be again expanded for small $\beta$, giving rise to an additional contribution to the Kaluza-Klein masses plus interaction terms between $X_0$ and $X_n$. Moreover, without loss of generality we can shift the $X_0$ mode as
\begin{equation}\label{eq:shift_X0}
    X_0\to X_0 - \frac{1}{\beta}\ln\left(2\pi r T\right).
\end{equation}
At the end we get a one-dimensional theory equipped with the constant einbein $u=\left(T/r\right)^2$ and consisting of three pieces:
\begin{equation}
    S=S(X_0)+S\left(\{X_n\}\right)+S\left(\{X_0,X_n\}\right).
\end{equation}
$S(X_0)$ is the action for a massless particle $X_0$ moving in a potential $V(X_0)=-\frac{r}{T}\mu e^{\beta X_0}$,
\begin{equation}
    S(X_0)=\int_0^1 d\tau \sqrt{u}\, \left(\frac{1}{2}u^{-1}\dot{X}_0^2-V(X_0)\right)
    =\int_0^1 d\tau \left(\frac{r}{2 T}\dot{X}_0^2+\mu e^{\beta X_0}\right).
\end{equation}
$S\left(\{X_n\}\right)$ contains an infinite tower of massive particles (the Kaluza-Klein modes), similar to what happens when compactifying the worldsheet of the Polyakov action in bosonic string theory \cite{Abel:2020gdi}. The masses are
\begin{equation}
    m^2_n=n^2+\frac{r\mu\beta^2}{T}, \hspace{4 mm} n\geq 1
\end{equation}
and
\begin{equation}\label{Xn_liouville}
\begin{gathered}
    S\left(\{X_n\}\right)=\sum_{n=1}^{+\infty}\int_0^1 d\tau \sqrt{u}\, \left(u^{-1} \lvert \dot{X}_n\rvert^2 +m^2_n\lvert X_n\rvert^2\right)\\ =\sum_{n=1}^{+\infty}\int_{0}^1 d\tau\, \bigg\{\frac{r}{T}\lvert \dot{X}_n\rvert^2 +\left(\frac{T}{r}n^2+\mu\beta^2\right)\lvert X_n\rvert^2\bigg\}.
\end{gathered}
\end{equation}
Finally, $S\left(\{X_0,X_n\}\right)$ involves interaction terms between the zero mode $X_0$ and the KK modes, describing a variety of (possibly virtual) processes:
\begin{equation}\label{eq:vertex}
\begin{gathered}
    S\left(\{X_0,X_n\}\right)=\frac{r}{T}\mu\beta^3\int_0^1 d\tau \sqrt{u}\, X_0 \lvert X_n\rvert^2 +\mathcal{O}(\beta^4)  =  \mu \beta^3\sum_{n=1}^{+\infty}\int_0^1 d\tau\, X_0 \lvert X_n\rvert^2+\mathcal{O}(\beta^4).
\end{gathered}
\end{equation}
The first non trivial interaction is a cubic vertex, describing the decay of $X_0$ into $X_n$ and $X^{\dagger}_n$ (or equivalently the annihilation of $X_n$ and $X^{\dagger}_n$ to form $X_0$), as shown in Fig.~\ref{fig1}. Keeping higher order terms in the expansion of $\exp{\beta X_0}$ allows for the existence of new interactions, which are however increasingly suppressed. 
\begin{figure}
\begin{center}
\begin{tikzpicture}
  \begin{feynman}
    \vertex (a) {\(X_0\)};
    \vertex [right=of a] (b);
    \vertex [above right=of b] (f1) {\(X_{n}\)};
    \vertex [below right=of b] (c) {\(X^{\dagger}_n\)};
    \diagram* {
      (a) -- [scalar] (b) -- [plain] (f1),  
      (b) -- [plain] (c), 
    }; $\hspace{30 mm}\sim \mu \beta^3$
  \end{feynman}
\end{tikzpicture}  
\end{center}
\caption{\label{fig1} A possible Feynman diagram for the interaction \eqref{eq:vertex}.}
\end{figure}
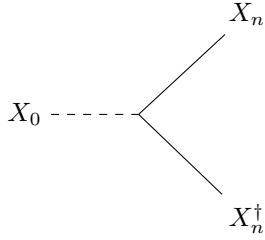
%%%%%%%%%%%%%%%%%%%%%%%%%%%
%%%%%%%%%%%%%%%%%%%%%%%%%%%
\subsection{The propagator of \texorpdfstring{$X_0$}{TEXT}}
In this subsection we compute the explicit expression of the propagator for $X_0$. The starting point is
\begin{equation}
    K_0(x_0,y_0)=\int \mathcal{D}X_0 \,e^{-S(X_0)},
\end{equation}
with boundary conditions $X_0(0)=x_0$, $X_0(1)=y_0$. It is not necessary to expand the exponential term of $S(X_0)$ because the path integral is exactly solvable. Indeed, path integrals of this kind were notably studied in \cite{Garay:1990re} in the context of quantum cosmology. We introduce two auxiliary fields $u$ and $v$ via the change of variables
\begin{equation}\label{eq:change_variables}
    u=\frac{1}{\beta}\cosh\left(\beta X_0\right), \hspace{4 mm} v=\frac{1}{\beta}\sinh\left(\beta X_0 \right).
\end{equation}
The action $S(X_0)$ simplifies to 
\begin{equation}\label{eq:actionX0}
S(X_0)=\int_0^1 d\tau\, \bigg\{\frac{r}{2 T}(-\dot{u}^2+\dot{v}^2)
+\mu \beta (u+v)\bigg\}.
\end{equation}
The classical equations of motions are 
\begin{equation}\label{eq:classical_eom}
    \ddot{u}=-\frac{\mu T\beta}{r}, \hspace{5 mm} \ddot{v}=\frac{\mu T\beta}{r}.
\end{equation}
Defining the boundary conditions
\begin{equation}
    u(0)=u', \hspace{2mm} u(1)=u'', \hspace{2mm}  v(0)=v',\hspace{2mm}  v(1)=v'',
\end{equation}
the solutions of \eqref{eq:classical_eom} are
\begin{equation}
\begin{split}
\Bar{u}(\tau)=-\frac{\mu T\beta}{2r}\tau^2+\left(\frac{\mu T\beta}{2r}+u''-u'\right)\tau+u',\\
\Bar{v}(\tau)=\frac{\mu T\beta}{2r}\tau^2+\left(-\frac{\mu T\beta}{2r}+v''-v'\right)\tau+v'.
\end{split}
\end{equation}
We now expand $u$ and $v$ around the classical solutions as
\begin{equation}\label{eq:expansions_uv}
\begin{split}
    u(\tau)=\Bar{u}(\tau)+\omega(\tau),\\
    v(\tau)=\Bar{v}(\tau)+\rho(\tau).
\end{split}
\end{equation}
and redefine the path integral measure as 
\begin{equation}
    \mathcal{N}\mathcal{D}\rho\mathcal{D}\omega,
\end{equation}
$\mathcal{N}$ being a normalization factor.
With the expansions \eqref{eq:expansions_uv} the action \eqref{eq:actionX0} becomes
\begin{equation}
    S(X_0)=S_{cl}+\frac{1}{2}\int_0^1 d\tau\, \frac{r}{T}\left(\dot{\rho}^2-\dot{\omega}^2\right),
\end{equation}
where
\begin{equation}\label{eq:action_classical}
\begin{gathered}
    S_{cl}=\int_0^1d\tau\,\left\{\frac{r}{2T}\left(\dot{\Bar{v}}^2-\dot{\Bar{u}}^2\right)+\mu\beta\left(\bar{u}+\bar{v}\right)\right\}\\
    =\frac{r}{2T} \left[(v'-v'')^2-(u'-u'')^2\right]
    +\frac{\mu\beta}{2} \left(u'+u''+v'+v''\right).
\end{gathered}
\end{equation}
The propagator is then
\begin{equation}
\begin{gathered}
    K_0(u',v',u'',v'')=\mathcal{N}\int \mathcal{D}\rho\mathcal{D}\omega\, e^{-S(X_0)}\\
    =\mathcal{N}e^{-S_{cl}}\int \mathcal{D}\rho\mathcal{D}\omega\,\exp{\left[-\frac{1}{2}\int_0^1 d\tau\, \frac{r}{T}\left(\dot{\rho}^2-\dot{\omega}^2\right)\right]}.
\end{gathered}
\end{equation}
The path integral over $\rho$ is identical to that for a free non-relativistic particle. Therefore
\begin{equation}
     K_0(u',v',u'',v'')=\sqrt{\frac{r}{2\pi T}}\mathcal{N}Ie^{-S_{cl}},
\end{equation}
where 
\begin{equation}
    I=\int \mathcal{D}\omega\, \exp\left[\frac{1}{2}\int_0^1 d\tau\,\frac{r}{T\beta^2}\dot{\omega}^2\right]
\end{equation}
is a clearly divergent integral. If we fix the normalization factor $\mathcal{N}$ by requiring that the Liouville propagator should recover the usual free non-relativistic particle propagator in the limits $\beta\to0$ and $\mu\to 0$, then we can absorb the divergence of $I$ (plus irrelevant constant factors like $\pi$) into $\mathcal{N}$:
\begin{equation}
    \mathcal{N}\sim I^{-1}.
\end{equation}
We now reinstate the $X_0$ dependence. Using \eqref{eq:change_variables} and recalling that $X(0)=x_0$, $X_0(1)=y_0$, we have
\begin{equation}
\begin{split}
    u'=\frac{1}{\beta}\cosh\left(\beta x_0\right), \hspace{3mm} u''=\frac{1}{\beta}\cosh\left(\beta y_0\right),\\
    v'=\frac{1}{\beta}\sinh\left(\beta x_0\right), \hspace{3mm} v''=\frac{1}{\beta}\sinh\left(\beta y_0\right).
\end{split}
\end{equation}
The classical action \eqref{eq:action_classical} becomes
\begin{equation}
    S_{cl}=\frac{r}{T\beta^2}\left\{\cosh\left[\beta(x_0-y_0)\right]-1\right\}
    +\frac{\mu}{2}\left(e^{\beta x_0}+e^{\beta y_0}\right)
\end{equation}
and the propagator is
\begin{equation}
  K_0\left(x_0,y_0\right) =\left(\frac{r}{T}\right)^{1/2}e^{-S_{cl}}.
\end{equation}
We can check that in the limits $\mu\to0$ and $\beta\to0$ the form of the propagator reduces to that of a free non-relativistic particle:
\begin{equation}
    K_0\left(x_0,y_0\right)\sim \left(\frac{r}{T}\right)^{1/2}e^{-\frac{r}{2T}(x_0-y_0)^2}.
\end{equation}
%%%%%%%%%%%%%%%%%%%%%%%%%%%%%%%
%%%%%%%%%%%%%%%%%%%%%%%%%%%%%%%
\subsection{Sewing the cylinder into a torus}
The next thing to compute is the amplitude for a single massive state $X_n$ to go from an initial boundary condition $x_n$ to a final $y_n$. This is given by the Mehler kernel:
\begin{equation}
K_n\left(x_n,y_n\right)=\int_{X_n(0)=x_n}^{X_n(1)=y_n} \mathcal{D}X_0\, e^{-S(X_n)}
=\frac{\omega_n r}{\pi T\sinh\omega_n}e^{-\frac{\omega_n r}{T}\frac{(\lvert x_n\rvert^2+\lvert y_n\rvert^2)\cosh\omega_n - 2\text{Re}(x_n\cdot y_n)}{\sinh\omega_n}},
\end{equation}
where
\begin{equation}
    \omega_n=\frac{T}{r}\sqrt{n^2+\frac{\mu r\beta^2}{T}}.
\end{equation}
We are now ready to compute the one-loop partition function by gluing together the far ends of the cylinder and by integrating over all possible kernels. A remark is in order. We set so far $A=0$. However, we have to take into account the ability of the fields to undergo a gauge transformation parametrised by a non-vanishing $A$ as they move around the loop. Hence, we identify $x_n \sim x_n e^{\frac{i nA}{r}}$.
The one-loop partition function is
\begin{equation}\label{eq:part_func_liouville}
\begin{gathered}
    Z(A,T)=\lim_{N\to+\infty}\prod_{n=1}^N\int_{-\infty}^{+\infty}dx_ndx^{\dagger}_n\, K_n\left(x_n,x_ne^{\frac{i nA}{r}}\right)\int_{-\infty}^{+\infty}dx_0\, K_0(x_0,x_0)\\
    =\left(\frac{r}{T}\right)^{1/2}\int_{-\infty}^{+\infty}dx_0\, e^{-\mu e^{\beta x_0}}\lim_{N\to +\infty}\prod_{n=1}^N\int_{-\infty}^{+\infty}dx_ndx^{\dagger}_n\,\frac{\omega_n r}{\pi T\sinh\omega_n} \\
    \times\exp\left\{-\frac{4\omega_n r}{T\sinh \omega_n}\lvert x_n \rvert^2\bigg\lvert\sin\left(\frac{nA}{2r}+\frac{i}{2}\omega_n\right) \bigg\rvert^2\right\}.
\end{gathered}
\end{equation}
The integral over $\lvert x_n\rvert$ is a standard Gaussian integral, while the other one over $x_0$ can be regularised as
\begin{equation}
    \int_{-\infty}^{+\infty}dx_0\, e^{-\mu e^{\beta x_0}}=-\frac{1}{\beta}\left(\gamma_E+\ln\mu \right),
\end{equation}
where $\gamma_E$ is the Euler-Mascheroni constant. 
Dropping irrelevant constant factors, the partition function is then
\begin{equation}
    Z(A,T)=-\frac{1}{\beta}\left(\frac{r}{T}\right)^{1/2}\ln\mu\lim_{N\to +\infty}\prod_{n=1}^{N}\frac{1}{\big\lvert\sin\left(\frac{nA}{2r}+\frac{i}{2}\omega_n\right) \big\rvert^2}.
\end{equation}
One can expand $\omega_n$ to any given order in $\beta$. At lowest order, 
\begin{equation}
\begin{gathered}
    \frac{1}{\big\lvert\sin\left(\frac{nA}{2r}+\frac{i}{2}\omega_n\right) \big\rvert^2}=\frac{1}{\big\lvert\sin\left(\frac{nA}{2r}+\frac{inT}{2r}\right) \big\rvert^2}+\mathcal{O}(\beta^2).
\end{gathered}
\end{equation}
Therefore, the infinite product can be easily computed by means of zeta function regularisation \cite{Abel:2020gdi}:
\begin{equation}
    \prod_{n=1}^{+\infty} \sin(2\pi n \tau)\stackrel{\zeta}{=}\sqrt{2}\eta(2\tau),
\end{equation}
where $\eta$ is the Dedekind eta function. 
The $T$ and $A$ parameters could now be rescaled as $T\to2\pi T$, $A\to2\pi A$, so to eliminate the factor $2\pi$ in the $\eta$ function and absorb it by an overall normalisation. However, we do not absorb $r$ into $T$ and $A$ as they represent physically distinct quantities: $r$ is the radius of the cylinder, $T$ is the Schwinger parameter. We want to continue keeping this distinction clear. 

The final result for the one-loop partition function at leading order in $\beta$ is then
\begin{equation}\label{eq:partitionfunc_liouville}
    Z(A,T,r)=-\frac{1}{\beta}\left(\frac{r}{T}\right)^{1/2}\frac{1}{\big\lvert\eta\left(\frac{A+iT}{r}\right) \big\rvert^2}\ln\mu+\mathcal{O}(\beta^2).
\end{equation}
We notice that the Schwinger parameter $T$ gets promoted to the full complex Teichm\"uller parameter $\frac{1}{r}(T+iA)$ of the torus. Strictly speaking, it is the additional presence of the gauge field $A$ that really turns \eqref{eq:part_func_liouville} into a one-loop toroidal diagram.

We would like now to comment about our new expression. First of all, the form of \eqref{eq:partitionfunc_liouville} is similar to the analogous expression obtained from the compactification of bosonic string theory on a cylinder \cite{Abel:2020gdi}. The difference is that here we have the additional $\beta$ and $\mu$ parameters of Liouville theory. The limits $\beta\to 0$ and $\mu\to 0$ seem to make $Z$ divergent, this is however not the case. In fact, in those limits $K_0$ becomes independent of $x_0$ and no integration over it is required. Hence, one correctly recovers the result of \cite{Abel:2020gdi} for $V(x)=0$.
Furthermore, using the following property of the Dedekind eta function
\begin{equation}
    \eta\left(\frac{i}{x}\right)=\sqrt{x}\eta\left(ix\right),
\end{equation}
it is easy to see that the partition function for the case $A=0$ is invariant under the transformation 
\begin{equation}\label{eq:Tduality}
    r\to\frac{T^2}{r},
\end{equation}
which is an example of T-duality. The radius $r=T$ is a fixed point of this transformation. It means that the compactified theory looks the same regardless whether we consider it at large or small radius of the internal circle. The partition function at the critical radius is independent of $T$:
\begin{equation}\label{eq:Z_fixedpoint}
   Z(A=0,T,r=T) = -\frac{1}{\beta}\frac{1}{\lvert\eta\left(i\right) \rvert^2}\ln\mu= -\frac{1}{\beta}\frac{4\pi^{\frac{3}{2}}}{\Gamma\left(\frac{1}{4}\right)^2}\ln\mu\sim -\frac{1}{\beta}\ln\mu.
\end{equation}
%%%%%%%%%%%%%%%%%%%%%%%%%%%%%%%
%%%%%%%%%%%%%%%%%%%%%%%%%%%%%%%
\subsection{Scattering amplitudes}\label{sec:amplitudes}
In order to construct amplitudes, we first need the worldline Green's function from the action \eqref{Xn_liouville}. It can be read off the two-point correlator
\begin{equation}
    \langle X^{\dagger} (v_1)X(v_2) \rangle = \delta_{nm} G_n\left(v_{12},T\right)
\end{equation}
where $v_{12}=\lvert v_1-v_2\rvert$ and the Green's function $G_n$ is the inverse of the operator
\begin{equation}
    L_n=-\frac{r}{T}\partial^2_{\tau}+\frac{T}{r}n^2+\mu\beta^2.
\end{equation}
Following \cite{Abel:2020gdi}, we restrict ourselves to line segments and circles and gauge fix $A=0$.
Over the real line, the Green's function is
\begin{equation}
    G_n\left(v_{12},T\right)=\frac{1}{4\pi \lvert s\rvert}e^{-2\pi v_{12}\lvert s\rvert\frac{T}{r}}
\end{equation}
while over a circle parametrised by $v\in[0,1]$ it is
\begin{equation}
    G^{\circ}_n\left(v_{12},T\right)=\frac{1}{4\pi \vert s\rvert}\sum_{k\in\mathcal{Z}}e^{-2\pi\lvert s\rvert \lvert k-v_{12}\rvert\frac{T}{r}},
\end{equation}
where $s^2=n^2+\frac{\mu\beta^2 r}{T}$.

The second ingredient that we need is the worldline vertex operator. Its gauge invariant expression is
\begin{equation}
    V=\int du\, \prod_m V_{u,m}=\int du\, \prod_m e^{ip\cdot e^{2\pi i m u X_m}}.
\end{equation}
A general one-loop scattering amplitude has then the form
\begin{equation}
    \mathcal{A}_m=\prod_{n_1,\cdots n_m \in \mathbb{Z}}\int dT\int u_1\cdots \int u_m \int v_1\cdots\int v_m\, \langle V^{\dagger}_{u_1,n_1}\left(k_1,v_1\right)V_{u_2,n_2}\left(k_2,v_2\right)\cdots\rangle Z(T,r).
\end{equation}
where $k_1,..,k_m$ refer to the momenta of the external states.

As an example, the amplitude involving two tachyon vertex operators on a circle with momenta $p$ and $q$ is
\begin{equation}
\begin{gathered}
    \mathcal{A}_2=\prod_{n_1,n_2\in\mathbb{Z}}\int dT\int du_1\int du_2\int dv\,\langle V^{\dagger}_{u_1,n_1}(p,0)V_{u_2,n_2}(q,v)\rangle Z(T,r)\\
    =\int dT\int du_1\int du_2\int dv\,e^{-pq \sum_{n\in \mathbb{Z}}e^{2\pi i n(u_1+u_2)G^{\circ}_n(v,T)}}Z\left(T,r\right).
\end{gathered}
\end{equation}
Another interesting example is the vacuum to vacuum amplitude. In this case, the only integration required is over $T$ in the interval $I=(0,+\infty)$\footnote{In principle one integrates over both $T$ and $A$, but we are focusing here on the case $A=0$.}. The measure receives an additional factor of $1/T$ coming from the sewing procedure, as the same circle arises from $\sim 2\pi T$ ways of sewing the dimensionally reduced line segment together \cite{Abel:2020gdi}. Moreover, let us apply the expression \eqref{eq:Z_fixedpoint} for $Z$ evaluated at the fixed point of the T-duality. We have
\begin{equation}\label{eq:our_amplitude}
    \mathcal{A}_0=-\frac{1}{\beta}\ln\mu \int_a^b \frac{dT}{T}=-\frac{1}{\beta}\ln\mu \left(\ln b-\ln a\right),
\end{equation}
The parameters $a$ and $b$ define the region of integration. Since $T=r$, they depend on some length scale $r'$. Moreover, they must satisfy some requirements. Firstly, they have to behave as $a\to 0$ and $b\to+\infty$ in some limiting regime, say $r'\to0$. Secondly, we require the ratio $b/a$ to be finite and symmetric under \eqref{eq:Tduality}. This means that 
\begin{equation}
    b=af\left(r'+\frac{T'}{r'}\right)
\end{equation}
for some function $f$. Here, $T'$ is a length scale that plays an analogous role as $T$. Without loss of generality we choose to fix $T'=1$.
If we assume a small compact radius $r'$, then for $r'\to 0$ the first requirement implies that $f$ must go to infinity faster than the way $a$ goes to zero.  This, along with dimensional analysis, fixes $f$ to an exponential form, hence
\begin{equation}
b=ae^{r'+\frac{1}{r'}}.
\end{equation}
Let us now perform a \textit{double compactification} \cite{Duff:1987bx}. In other words,  we assume that, in addition to the worldsheet, also the target space is compact. Furthermore, we choose $r'$ to be the radius of the target space. Hence, the amplitude is
\begin{equation}
    \mathcal{A}_0 = -\frac{1}{\beta}\left(r'+\frac{1}{r'}\right)\ln \mu,
\end{equation}
 Our result is identical to the one obtained in \cite{Gross:1990ub, Nakayama:2004vk} for the vacuum to vacuum amplitude (called partition function in these references) of Liouville theory with a compact target space. Our route to the result is advantageous since it does not involve ghosts and avoids various complications of matrix models. What we have done is to dimensionally reduce Liouville theory to a worldline, then compute the leading term of the small $\beta$ expansion of the partition function for $A=0$ and $r=T$, and finally integrate over $T$. 
 
We conclude this section by noticing that one could obviously compute a generic amplitude (i.e. without assuming a compact target space) by starting with the full \eqref{eq:partitionfunc_liouville} and integrating over $T$ and $A$. This would likely involve the use of the Rankin-Selberg technique \cite{Abel:2021tyt}. 
%%%%%%%%%%%%%%%%%%%%%%%%%%%%%%%%
%%%%%%%%%%%%%%%%%%%%%%%%%%%%%%%%
\section{2D quantum field theories as worldline theories}\label{sec:section2}
The analysis of the previous section relied on the particular exponential form of the Liouville potential. The circle compactification was carried at low coupling, and information about the cylinder amplitude was extracted using standard tools. It is then natural to ask which QFTs allow for a similar treatment. Our procedure is not immediately useful for any other potential, since we crucially relied upon the basic fact that $e^{(x+y)}=e^x e^y$. There is one other particular advantage in working on the exponential potential, though: any other Fourier expandable potential can be reduced to it, in the following sense. Consider the field theory
\begin{equation}\label{eq:action_soliton}
    S=\frac{1}{4\pi}\int_\Sigma d\sigma^2\sqrt{g}\, \left(g^{ab}\partial_aX\partial_b X+4\pi V(X)\right),
\end{equation}
where $X$ is a 2D scalar field and 
\begin{equation}
    V(X)=-\frac{2a^2}{\cosh^2{aX}}, \hspace{3mm} a>0.
\end{equation}
is a generalisation of the P\"{o}schl-Teller quantum potential, also known as one-soliton potential \cite{pupasov2005}. The role of the coupling constant is interpreted by $a$. Applying a Fourier transform, the potential is written 
\begin{equation}
    V(X)=-a^2 \int_{-\infty}^{+\infty}dk\, \frac{k}{\sinh\left(\frac{\pi k}{2}\right)}e^{-iakX}.
\end{equation}
Now we carry out the same compactification procedure of Sec.\ref{sec:section1}. Using \eqref{eq:integral_exponential} and \eqref{eq:shift_X0} with $\beta\to-iak$, we find that the action \eqref{eq:action_soliton} splits into three pieces:
\begin{equation}
    S=S(X_0)+S\left(\{X_n\}\right)+S\left(\{X_0,X_n\}\right),
\end{equation}
where
\begin{equation}
   S(X_0)=\int_0^1 d\tau\, \left(\frac{r}{2T}\dot{X}_0-\frac{2a^2}{\cosh^2\left(aX_0\right)}\right),
\end{equation}
\begin{equation}\label{eq:xn_soliton}
    S\left(\{X_n\}\right)=\sum_{n=1}^{+\infty}\int_0^1 d\tau\,  \left\{\frac{r}{T}\lvert \dot{X}_n\rvert^2+\left(\frac{T}{r}n^2+4a^4\right)\lvert X_n\rvert^2\right\},
\end{equation}
\begin{equation}
    S\left(\{X_0,X_n\}\right)=\mathcal{O}(a^5).
\end{equation}
Since \eqref{eq:xn_soliton} and \eqref{Xn_liouville} have the same structure, the amplitude for a single massive state $X_n$ to go from an initial boundary condition $x_n$ to a final $y_n$ is
\begin{equation}
K_n\left(x_n,y_n\right)
=\frac{\omega_n r}{\pi T\sinh\omega_n}e^{-\frac{\omega_n r}{T}\frac{(\lvert x_n\rvert^2+\lvert y_n\rvert^2)\cosh\omega_n - 2\text{Re}(x_n\cdot y_n)}{\sinh\omega_n}},
\end{equation}
where
\begin{equation}
    \omega_n=\frac{T}{r}\sqrt{n^2+4a^4\frac{r}{T}}.
\end{equation} 
The propagator for $X_0$ is \cite{pupasov2005} 
\begin{equation}
    K_0(x_0,y_0)=\sqrt{\frac{r}{4\pi T}}e^{-\frac{1}{4}(x_0-y_0)^2}+\sqrt{\frac{2r}{T}}e^{\frac{a^2T}{2r}}\frac{\text{erf}_{+}\left(a\sqrt{\frac{T}{2r}}\right)+\text{erf}_{-}\left(a\sqrt{\frac{T}{2r}}\right)}{4\cosh\left(a\sqrt{\frac{T}{2r}x_0}\right)\cosh\left(a\sqrt{\frac{T}{2r}}y_0\right)},
\end{equation}
with
\begin{equation}
    \text{erf}_{\pm}(a)=\text{erf}\left(a\pm\frac{i}{2}(x_0-y_0)\right).
\end{equation}
The one-loop partition function is 
\begin{equation}
    Z(A,T,r)=\lim_{N\to+\infty}\prod_{n=1}^{+\infty}\int_{-\infty}^{+\infty}dx_nd^{\dagger}x_n\, K_n\left(x_n,x_ne^{\frac{inA}{r}}\right)\int_{-\infty}^{+\infty}dx_0\, K_0\left(x_0,x_0\right).
\end{equation}
A remark is in order. Without the $a$ dependent part of $K_0$, there would be no $x_0$ in $Z$, consequently we would not have to integrate over it. Hence, the integral over $x_0$ is intended to be over the $a$ dependent part only, thereby avoiding an infinite term. 
Following the calculations of Section \ref{sec:section1}, recalling that $\text{erf}(x)=\frac{2x}{\sqrt{\pi}}+\mathcal{O}(x^2)$ and neglecting irrelevant constant factors, at first order in $a$ we find
\begin{equation}\label{eq:partfunc_soliton}
    Z(A,T,r)=\left(\frac{r}{T}\right)^{\frac{1}{2}}\frac{1+a}{\lvert \eta\left(\frac{A+iT}{r}\right) \rvert^2}+\mathcal{O}(a^4).
\end{equation}
The structure of the partition function is similar to that found for Liouville theory. As such, for $A=0$ it presents the same T-duality $r\to T^2/r$. The radius $r=T$ is a fixed point of the transformation and with this choice the partition function simplifies and becomes independent of $T$. The limit $a\to 0$ recovers the result of $V(x)=0$ obtained in \cite{Abel:2020gdi}.
The procedure that we carried out can be applied to any two-dimensional field theory, provided that the potential can be written in terms of its Fourier transform. If this is the case, then the exponential term can be compactified according to \eqref{eq:integral_exponential}. The upshot of this procedure is that one can relate any 2D field theory on a cylinder, regardless of the potential function defining it, to an ordinary quantum mechanical system, at least at small coupling. The universal behavior of the torus partition function of such theories is discussed in the following section.
%%%%%%%%%%%%%%%%%%%%%%%%%%%
%%%%%%%%%%%%%%%%%%%%%%%%%%%
\section{Universality of the torus partition function at small coupling}\label{sec:section3}
In this section we prove a general result involving cylinder quantum field theories at small coupling. The proof closely follows the steps illustrated in the previous sections. 
\begin{theorem}
Let $\mathcal{F}[\Sigma;\epsilon]$ be a two-dimensional quantum field theory for a scalar field $X$, with coupling constant $\epsilon$ and defined on a cylinder $\Sigma$ of radius $r$. Let $T$ be the Schwinger parameter of the cylinder, $A$ a gauge field describing its twist and $V_\epsilon$ a potential such that its Fourier transform exists. Moreover, let $\mathcal{F}[\Sigma;\epsilon]$ have the Lagrangian description 
\begin{equation}
S=\frac{1}{4\pi}\int_{\Sigma}d^2\sigma \sqrt{g}\,\left(g^{ab}\partial_a X\partial_b X+4\pi V_\epsilon(X)\right).
\end{equation}
Then the one-loop partition function obtained from the corresponding worldline theory by sewing the cylinder into a torus is, at leading order in the coupling,
\begin{equation}
    Z(A,T,r)=\left(\frac{r}{T}\right)^{\frac{1}{2}}\frac{1+\epsilon^n}{\lvert \eta\left(\frac{A+iT}{r}\right)\rvert^2},
\end{equation}
for some power $n$ of the coupling. 
In particular, $Z(A=0,T,r)$ is invariant under the symmetry transformation
\begin{equation}
    r\to\frac{T^2}{r}
\end{equation}
\end{theorem}
\begin{proof}
The compactification of the theory on the cylinder, parametrised by $\sigma_1\in[0,2\pi r]$ and $\sigma_2\equiv\tau\in[0,1]$, begins with expanding $X$ as
\begin{equation}
    X(\sigma_1,\sigma_2)=\sum_{n\in \mathbb{Z}}X_n(\sigma_2)e^{ \frac{i n\sigma_1}{r}},
\end{equation}
and gauge fixing the metric to
\begin{equation}
    g_{ab}=\begin{pmatrix}1 & 0\\ 0 & T^2
    \end{pmatrix}.
\end{equation}
The potential $V_{\epsilon}$ can be written as
\begin{equation}
    V_{\epsilon}(X)=\frac{1}{\sqrt{2\pi}}\int_{-\infty}^{+\infty}dk\, \Tilde{V}_{\epsilon}(k) e^{-i \epsilon k X}
\end{equation}
where $\Tilde{V}_{\epsilon}(k)$ is the Fourier transform of $V_{\epsilon}(X)$. 
In Subsection \ref{sec:liouville_cylinder} we showed that (see \eqref{eq:integral_exponential})
\begin{equation}
    \int_0^{2\pi r} d\sigma_1 \, e^{-i\epsilon k X}= (2\pi r) e^{-i\epsilon k X_0}\left(1-\epsilon^2 k^2\sum_{n=1}^{+\infty}\lvert X_n\rvert^2+\mathcal{O}\left(\epsilon^3\right)\right).
\end{equation}
After a shift $X_0\to X_0-\frac{1}{-iak}\ln(2\pi r T)$, $S$ splits into three pieces:
\begin{equation}
    S=S\left(X_0\right)+S\left({X_n}\right)+S\left({X_0,X_n}\right),
\end{equation}
where 
\begin{equation}
\begin{gathered}
    S\left(X_0\right)=\int_0^1 d\tau\, \left(\frac{r}{2T}\dot{X}_0+\frac{1}{\sqrt{2\pi}}\int_{-\infty}^{+\infty}dk\,\Tilde{V}_{\epsilon}(k)e^{-i\epsilon k X_0}\right)\\
    =\int_0^1 d\tau\, \left(\frac{r}{2T}\dot{X}_0+V(X_0)\right)
\end{gathered}
\end{equation}
and
\begin{equation}
    S\left({X_n}\right)=\sum_{n=1}^{+\infty}\int_0^1d\tau\, \left\{\frac{r}{T}\lvert \dot{X_n}\rvert^2+m_n^2\lvert X_n\rvert^2\right\}.
\end{equation}
The masses $m_n$ of the KK modes are 
\begin{equation}
    m_n^2=\frac{T}{r}n^2+m_\epsilon
\end{equation}
where $m_\epsilon$ is some constant depending on the exact nature of $V_{\epsilon}$.
The  term $S\left({X_0,X_n}\right)$ involves interactions between the zero mode and the KK modes and its expression is not needed in this proof. Since $m_n^2$ is a constant term, the propagator for $X_n$ is given by the Mehler kernel and has the structure
\begin{equation}
K_n\left(x_n,y_n\right)
=\frac{\omega_n r}{\pi T\sinh\omega_n}e^{-\frac{\omega_n r}{T}\frac{(\lvert x_n\rvert^2+\lvert y_n\rvert^2)\cosh\omega_n - 2\text{Re}(x_n\cdot y_n)}{\sinh\omega_n}},
\end{equation}
where
\begin{equation}
    \omega_n=\frac{T}{r}\sqrt{n^2+m_\epsilon\frac{r}{T}}.
\end{equation} 
Now, the exact expression of the propagator for $X_0$ depends on the form of the potential. However, in the limit $\epsilon\to 0$ one must recover the free particle propagator. Therefore, we can write
\begin{equation}
    K_0(x_0,y_0)=\left(\frac{r}{4\pi T}\right)^{\frac{1}{2}}e^{-\frac{1}{4}(x_0-y_0)^2}+\Pi_{\epsilon}(x_0,y_0),
\end{equation}
where 
$\Pi_{\epsilon}(x_0,y_0)$ is the $\epsilon$-dependent part of the propagator. Since $T$ has the dimension of a length, the propagator $K_0$ is dimensionless, as it is clear from the first component. Thus, regardless of what the exact expression of $\Pi_{\epsilon}$ is, it must contain the combination $\left(r/T\right)^{1/2}$.  The one-loop partition function, obtained by integrating over all kernels and by taking into account the gauge field $A$, is
\begin{equation}
\begin{gathered}
    Z(A,T,r)=\lim_{N\to+\infty}\prod_{n=1}^{+\infty}\int_{-\infty}^{+\infty}dx_nd^{\dagger}x_n\, K_n\left(x_n,x_ne^{\frac{inA}{r}}\right)\left\{\left(\frac{r}{4\pi T}\right)^{\frac{1}{2}}+\int_{-\infty}^{+\infty}dx_0\, \Pi_{\epsilon}\left(x_0,x_0\right)\right\}.
\end{gathered}
\end{equation}
The integral of $K_n$, at leading order in the coupling, yields a factor of $\lvert \eta\left(\frac{T+iA}{r}\right)\rvert^{-2}$. The integral of $\Pi_{\epsilon}$ gives some contribution depending on the first non-trivial power $n$ of the coupling, from which we can factorise out the factor $\left(r/T\right)^{1/2}$. Hence, neglecting constant factors, the one-loop partition function at leading order, is
\begin{equation}
    Z(A,T,r)=\left(\frac{r}{T}\right)^{\frac{1}{2}}\frac{1+\epsilon^n}{\lvert \eta\left(\frac{A+iT}{r}\right)\rvert^2},
\end{equation}
which is the desired result. 
\end{proof}
As a corollary, the fact that $Z$ is at leading order independent of $V$ implies that amplitudes computed from different theories have a similar structure. 
%%%%%%%%%%%%%%%%%%%%%%%%
%%%%%%%%%%%%%%%%%%%%%%%%
\section{Summary and Outlook}
In this paper we studied the relationship between the Kaluza-Klein reduction procedure and the worldline formalism. We showed that any 2D scalar field theory compactified on a cylinder and with a Fourier expandable potential $V$ is equivalent, in the small coupling limit, to a 1D theory involving a massless particle satisfying the Schr\"oedinger equation with a potential $V$ and an infinite tower of free massive KK modes. As soon as one moves slightly away from the deep IR region, one can probe interactions between the zero mode and the KK modes. The new interactions are proportional to higher powers of the coupling constant, and hence are increasingly suppressed.  
We illustrated our procedure by first considering the well-studied Liouville field theory. We computed the one-loop partition function for this theory by deforming the cylinder into a torus and by integrating over all regularised kernels. We found that the partition function so obtained is invariant under a certain transformation which maps the radius to the cylinder to its inverse, and rescales it by the square of the Schwinger parameter of the cylinder. We computed the general expressions for one-loop amplitudes, obtained by integrating the partition function with worldline vertex operators. Moreover, we performed a double compactification, by additionally considering a compact target space. In this setting, the vacuum to vacuum amplitude recovers the one obtained in the literature by matrix models techniques. After this proof-of-concept calculation, we moved on to consider a theory with a one-soliton, or P\"{o}schl-Teller, potential and repeated the same steps. The one-loop partition function presents the same form and symmetry property as that for the Liouville case. We showed that this is indeed a universal feature of any cylinder field theory. 

Several possible directions stem from this paper. An idea for a future project would be to compute the amplitudes coming from the partition functions \eqref{eq:part_func_liouville}, \eqref{eq:partfunc_soliton}  by explicitly integrating over $A$ and $T$. The integrals could be performed by expanding the eta function in power series \cite{heim2019} and by applying the Rankin-Selberg technique \cite{Abel:2021tyt}. 
Another interesting idea would be to try to derive higher order partition functions, namely those defined on surfaces of genus greater than one. The problem in this case becomes more complicated. The difficulty at say two-loop order would be that one does not only have to determine the Green's function on the two-loop topology, but it is not even clear if this is the correct approach. In fact, one would have to figure out how to include the presence of the internal vertex operators and it is not obvious whether this is just a question of matching propagators at the vertices. 
The method illustrated in this paper could also be generalised to a generic $D$ dimensional field theory. We expect to address some of these open problems in future works.

\section*{Acknowledgements}
The authors thank the anonymous referees for their valuable comments which improved the quality of the article.

\bibliography{references}
\bibliographystyle{JHEP}

\end{document}